\documentclass[sigconf]{acmart}
\usepackage{tabularx}
\usepackage{booktabs}
\usepackage{algorithm}
\usepackage{multirow}
\usepackage{enumitem}
\usepackage{bm}
\usepackage{xspace}
\usepackage{marvosym}
\usepackage[normalem]{ulem}
\usepackage{makecell} 
\usepackage{graphicx} 
\usepackage{subcaption} 
\usepackage{pifont}
\usepackage{tikz} 
\usepackage{booktabs}

\usepackage{bbding}
\usepackage{algorithm}  
\usepackage{algpseudocode}
\usepackage{amsthm}
\newtheorem{assumption}{Assumption}
\newtheorem{proposition}{Proposition}
\AtBeginDocument{%
  }

\setcopyright{acmlicensed}
\copyrightyear{2018}
\acmYear{2018}
\acmDOI{XXXXXXX.XXXXXXX}
\acmConference[Conference acronym 'XX]{Make sure to enter the correct
  conference title from your rights confirmation email}{June 03--05,
  2018}{Woodstock, NY}
\acmISBN{978-1-4503-XXXX-X/2018/06}

\begin{document}


\title{Causality Enhancement for Cross-Domain Recommendation}
\author{Zhibo Wu, Yunfan Wu, Lin Jiang, Ping Yang, Yao Hu}
\email{{wuzhibo, wuyunfan, jianglin, jiadi, xiahou}@xiaohongshu.com}
\affiliation{%
  \institution{Xiaohongshu  Co., Ltd}
  \country{Beijing, China}
}

\renewcommand{\shortauthors}{Trovato et al.}
\newcommand{\ds}{CLM}
\newcommand{\ms}{DCMM}
\newcommand{\es}{CEM}
\newcommand{\fr}{CE-CDR}
\newcommand{\data}{Causality Labeling Module}
\newcommand{\model}{Direct Causality Modeling Module}
\newcommand{\enhance}{Causality Enhancement Module}
\setlength{\textfloatsep}{3.3pt plus 1pt minus 1pt}
\setlength{\intextsep}{3.3pt plus 1pt minus 1pt}
\setlength{\floatsep}{3.3pt plus 1pt minus 1pt}
\setlength{\abovecaptionskip}{3.3pt plus 1pt minus 1pt}
\setlength{\belowcaptionskip}{3.3pt plus 1pt minus 1pt}
\setlength{\abovedisplayskip}{3.3pt plus 1pt minus 1pt} 
\setlength{\belowdisplayskip}{3.3pt plus 1pt minus 1pt} 

\begin{abstract}
Cross-domain recommendation forms a crucial component in recommendation systems. It leverages auxiliary information through source domain tasks or features to enhance target domain recommendations.
However, incorporating inconsistent source domain tasks may result in insufficient cross-domain modeling or negative transfer. 
While incorporating source domain features without considering the underlying causal relationships may limit their contribution to final predictions. 
Thus, a natural idea is to directly train a cross-domain representation on a causality-labeled dataset from the source to target domain. Yet this direction has been rarely explored, as identifying unbiased real causal labels is highly challenging in real-world scenarios.
In this work, we attempt to take a first step in this direction by proposing a causality-enhanced framework, named CE-CDR. Specifically, we first reformulate the cross-domain recommendation as a causal graph for principled guidance. We then construct a causality-aware dataset heuristically. Subsequently, we derive a theoretically unbiased Partial Label Causal Loss to generalize beyond the biased causality-aware dataset to unseen cross-domain patterns, yielding an enriched cross-domain representation, which is then fed into the target model to enhance target-domain recommendations.
Theoretical and empirical analyses, as well as extensive experiments, demonstrate the rationality and effectiveness of CE-CDR and its general applicability as a model-agnostic plugin.
Moreover, it has been deployed in production since April 2025, showing its practical value in real-world applications.
\end{abstract}




\received{20 February 2007}
\received[revised]{12 March 2009}
\received[accepted]{5 June 2009}

\maketitle
\begin{figure}[t]
    \centering
    \subfloat[\small Illustration of a typical cross-domain causal relationship.]{
    \includegraphics[width=0.42\textwidth]{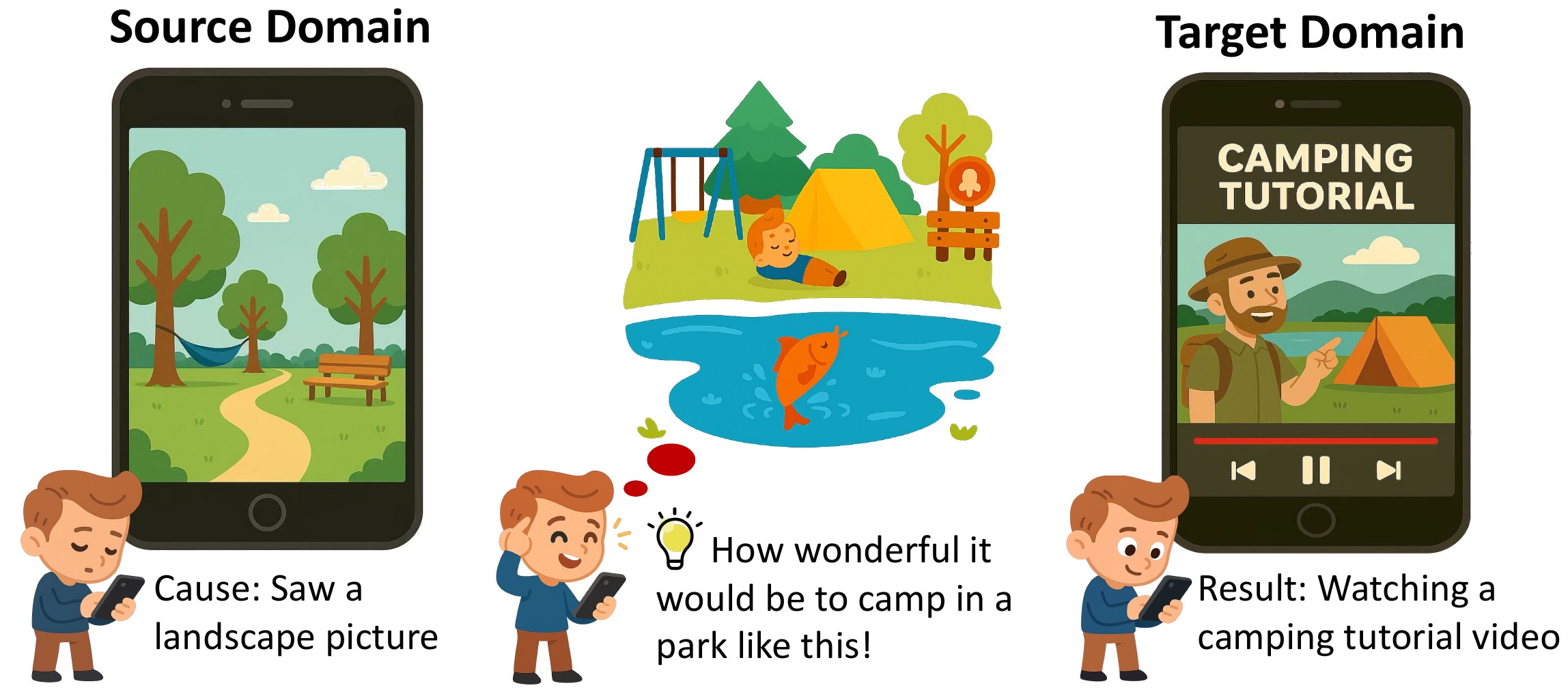}
        \label{fig:multi-task}
    }
    \par
    \subfloat[\small Causal graph of cross-domain recommendation.]{
    \includegraphics[width=0.42\textwidth]{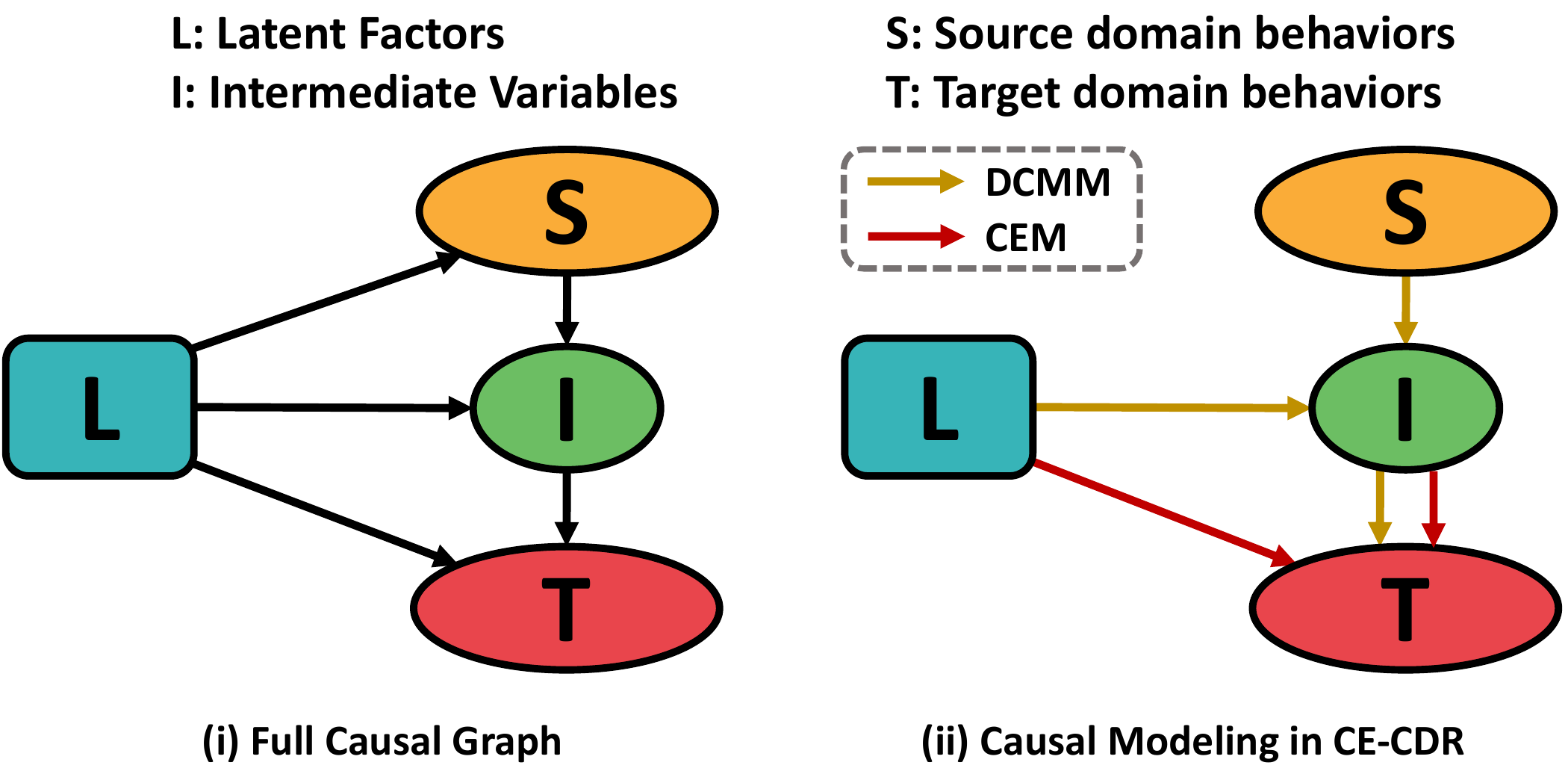}
        \label{fig:submerge}
    }
    \caption{\small Cross-domain causal illustration and causal graph.}
\end{figure}
\section{Introduction}\label{sec: introduction}

Recommender systems have been widely applied~\cite{deldjoo2023review, wu2023personalized, gao2022kuairec} to provide personalized recommendations. 
However, classical recommender systems that rely on interactions within a specific platform or domain often suffer from the data sparsity problem~\cite{idrissi2020systematic}, which hinders the comprehensive user preference modeling. 

To alleviate this problem, Cross-Domain Recommendation (CDR) has been proposed, leveraging user data from the source domain to assist recommendation in the target domain~\cite{zhu2021cross, khan2017cross, zang2022survey}. 
A typical example of cross-domain pattern is depicted in Figure~\ref{fig:multi-task}, where a user who saw a landscape picture in the note section (source domain), got inspired to camping, and then watched a camping tutorial video in the video section (target domain). 
Hence, by utilizing behavioral data from the notes domain, recommendation quality in the video domain can be improved. 

To better characterize the underlying mechanism of CDR, we reformulate the task as a causal graph in Figure~\ref{fig:submerge}. 
Latent factors $\bm{L}$ denote user preferences, system-level factors, etc, and intermediate variables $\bm{I}$ denote valuable cross-domain transferable information.
The goal of single-domain recommendation is learning a function: $\bm{L} \rightarrow \bm{T}$, while CDR aims to improve the prediction of $\bm{T}$ by leveraging $\bm{S}$ from the source domain. Following this, we then revisit existing CDR methods from a causal perspective.

CDR methods leverage source domain data in two primary ways.
Multi-task learning approaches simultaneously optimize recommendation tasks in both domains, i.e., modeling two functions: $\bm{L}\rightarrow \bm{S}$ and $\bm{L}\rightarrow \bm{T}$.
They facilitate cross-domain knowledge transfer through the use of shared model parameters~\cite{singh2008relational, lin2024mixed, chen2024improving} or user representations~\cite{hu2018conet, li2020ddtcdr, li2021dual, zhao2024discerning}.
However, due to inconsistent domain tasks~\cite{yu2020gradient}, shared parameters or representations may fail to capture valuable cross-domain information.
Especially when domains exhibit different user behavior patterns, the source domain task may even induce negative transfer~\cite{liu2019loss, li2023one}.

Another line of research incorporates users' source domain behaviors as additional input features to make recommendations in target domain~\cite{huang2024exit, loni2014cross}, i.e., modeling a function: $(\bm{L}, \bm{S})\rightarrow \bm{T}$.
Typically, attention mechanisms~\cite{ouyang2020minet} and graph neural networks~\cite{zhao2023cross, song2024mitigating} are employed to integrate user behavior across domains.
However, they still struggle to make full use of variables $S$, due to the sparsity of causal relationships between source and target domains~\cite{dai2021poso,zhang2023cold}.
Specifically, user preferences in target domain may not necessarily be influenced by source-domain behaviors.
Thus, indiscriminately training recommendation models to predict all target domain items from source domain behaviors could hinder the effectiveness of causal learning from source domain to target domain.

To overcome the challenge of task inconsistency and causality sparsity, a natural idea is to directly train a cross-domain representation $\bm{I}$ on a causal dataset from source to target domain. Yet this direction has been rarely explored, as obtaining unbiased real causal labels is highly challenging in real-world scenarios. In this work, we attempt to take a first step in this direction. We propose a causality-enhanced
framework, named \textbf{\fr}, that directly models causal relationships from both data and model perspectives. 
From the data side, \data ~(\ds) heuristically constructs similarity-driven causal supervision signals under psychological assumptions for downstream model to extract $\bm{I}$. 
From the model side, \model ~(\ms) trains a causal model  $(\bm{L}, \bm{S}) \rightarrow \bm{I} \rightarrow \bm{T}$ to generalize unseen cross-domain patterns beyond labeling mechanisms under a theoretically unbiased Partial Label Causal Loss (PLCL). It learns a cross-domain representation, which is then fed into the \enhance ~(\es) to enhance the target-domain recommendation.


The main contributions of this work are: 
\begin{itemize}[leftmargin=*, topsep=1pt, partopsep=0pt, itemsep=0pt, parsep=0pt]
    \item We reformulate the task of cross-domain recommendation as a causal graph for principled guidance and revisit existing CDR methods from a causal perspective.
    \item To the best of our knowledge, we are the first to propose a framework that directly models cross-domain causality from both data and model perspectives. Specifically, we heuristically construct a causality-aware dataset under psychological assumptions and derive a theoretically unbiased loss for generalized cross-domain pattern learning.
    \item Both theoretical and empirical analyses, together with experimental results, validate the rationality and effectiveness of the causality-aware dataset construction. It also demonstrates the ability of PLCL to generalize to unseen causal patterns.
    \item We conduct comprehensive experiments and online A/B tests, demonstrating its effectiveness and its general applicability as a model-agnostic plugin. Furthermore, it has been deployed in production, showing its practical value.
\end{itemize}

\section{Related Works}

\subsection{Cross-Domain Recommendation}
Cross-domain recommendation (CDR)~\cite{cantador2015cross} has been introduced to alleviate the data sparsity problem. 
Existing approaches can be classified into two main categories. 
Multi-task learning methods simultaneously optimize recommendation tasks across domains. 
They facilitate cross-domain knowledge transfer through shared model parameters or representations. 
Shared model parameters include user embeddings~\cite{singh2008relational}, item embeddings~\cite{gao2019cross}, group embeddings~\cite{lin2024mixed}, aspect embeddings~\cite{zhao2020catn}, or user preference generators~\cite{chen2024improving}. 
In contrast, shared representations establish mapping functions~\cite{hu2018conet, li2020ddtcdr, zhu2019dtcdr, fu2019deeply} or distance constraints~\cite{zhao2024discerning, li2021dual} between representations across domains to learn shared preferences. 
However, due to the discrepancies in domain-specific recommendation tasks~\cite{liu2019loss, li2023one}, shared parameters or representations may fail to capture valuable cross-domain information. 
Alternative CDR methods treat source domain user behaviors as additional input features~\cite{huang2024exit, loni2014cross}. 
Attention~\cite{ouyang2020minet, hu2018mtnet} or graph networks~\cite{zhao2023cross, song2024mitigating, zhao2019cross, liu2020cross} are typically employed to aggregate information from source domain interactions.
However, these methods overlook the causality sparsity during training, hindering their
contribution to final predictions~\cite{dai2021poso,zhang2023cold}. Specifically, user preferences in the target domain are not necessarily influenced by their source-domain behaviors, as users may exhibit distinct interests across domains.

\subsection{Causal Recommendation}
While recommendation systems primarily rely on co-occurrence patterns, such as feature-behavior associations, real-world decision-making is fundamentally governed by causality rather than co-occurrence. 
As a result, causal inference has been introduced into recommendation~\cite{luo2024survey}, addressing three key challenges. 
The first is data bias, including exposure bias~\cite{zhang2021causal, wei2021model}, popularity bias~\cite{schnabel2016recommendations, wang2019doubly, saito2020unbiased, chen2021autodebias, wang2022unbiased}, and conformity bias~\cite{zheng2021disentangling}. 
Another challenge pertains to missing data~\cite{wang2021counterfactual, yang2021top, wang2022causal, he2022causpref}, where the available data fails to capture users' comprehensive user preferences, even without bias. 
In addition, causal inference has also been employed to improve explainability~\cite{si2022model, tan2021counterfactual, xian2019reinforcement}, diversity~\cite{wang2022user, xu2022dynamic}, and fairness~\cite{shao2024average, huang2022achieving}. 
Among them, counterfactual inference~\cite{pearl2009causality} and inverse propensity weighting (IPW)~\cite{hirano2003efficient} are most widely adopted for causal inference.
As for relatively underexplored causality-based cross-domain recommendation, existing methods mainly focus on enhancing disentangled user representation learning~\cite{menglin2024c2dr, li2024causalcdr, du2024identifiability} and out-of-distribution generalization~\cite{zhang2024transferring, li2024cross}.
The direct modeling of cross-domain causality has been rarely explored. In this work, we attempt to take a first step in this direction. 
The causality is considered throughout both the data construction and model learning phases.

\begin{figure*}[t!]
    \centering
    \includegraphics[width=0.95\linewidth]{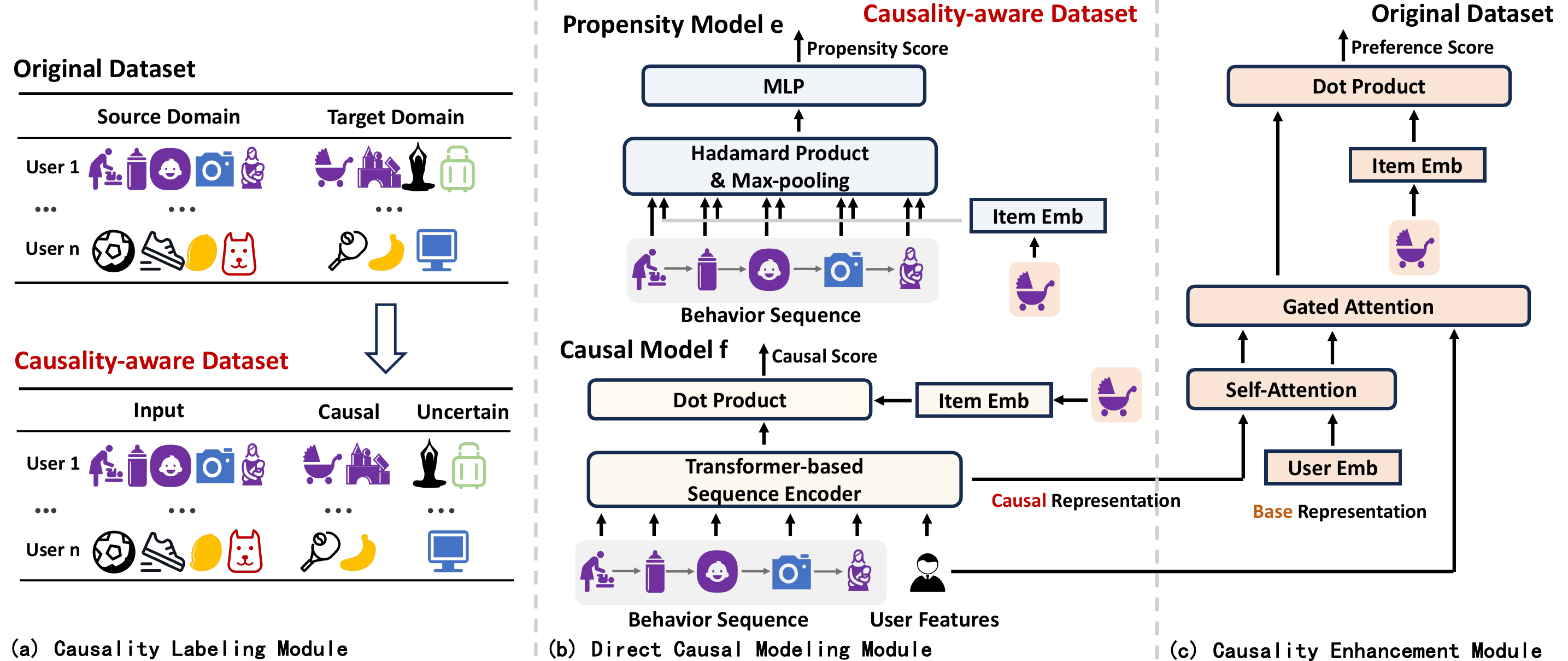}
    \caption{\small Overall architecture of \fr. (a) \ds: constructing causality-aware source-target behavior pairs; (b) \ms: learning how source domain behaviors influence target domain preferences; (c) \es: enhancing the target model for comprehensive user preference modeling.}
    \label{fig: overall}
\end{figure*}

\section{Methodology}
\subsection{Overview}
\subsubsection{Problem Formulation}
We consider a cross-domain recommendation scenario involving a source domain $s$ and a target domain $t$, both of which share a unified user set $\mathcal{U}$.  
Each domain has its own item set: $\mathcal{I}_s$ for the source domain and $\mathcal{I}_t$ for the target domain.  
For clarity, we denote items from each domain as $i_s \in \mathcal{I}_s$ and $i_t \in \mathcal{I}_t$, respectively.  
Each user $u \in \mathcal{U}$ has interaction sequences in either or both domains, represented as $\mathcal{S}_s^u = (i_s^{u,1}, i_s^{u,2}, \dots)$ in the source domain and $\mathcal{S}_t^u = (i_t^{u,1}, i_t^{u,2}, \dots)$ in the target domain.  
The goal of cross-domain recommendation is to leverage knowledge from the source domain $s$ to improve recommendation accuracy in the target domain $t$, i.e., recommending items $i_t \in \mathcal{I}_t$ that best match users’ preferences and maximize their engagement.

\subsubsection{Overall Architecture}
The overall architecture of \fr{} is depicted in Figure \ref{fig: overall}, which demonstrates the pipeline comprising \data{}~(\ds{}), \model~(\ms{}), and \enhance~(\es{}). 
\ds{} is responsible for constructing high-quality, causality-aware behavior pairs. 
\ms{} aims to generalize cross-domain causal patterns from a limited set of causal samples guided by the Partial Label Causal Loss. 
Finally, \es{} incorporates the learned causal representations into the target domain model in an adaptive manner.

\subsection{\data}\label{sec: Causality Labeling Module}
\begin{figure}[t]
    \centering
    \includegraphics[width=0.45\textwidth]{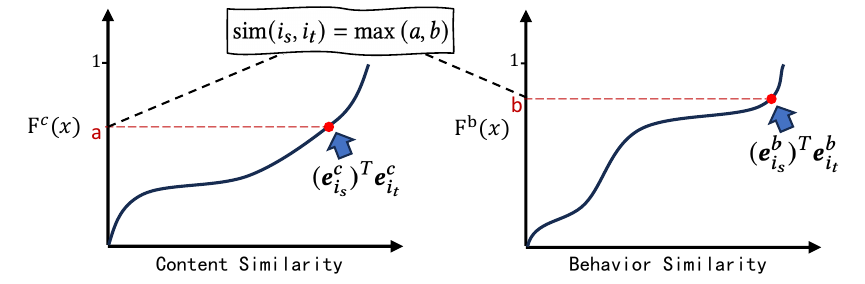}   
    \caption{Illustration of similarity calibration and fusion.} 
    \label{fig:calibration}
\end{figure}
In this section, we describe how \ds{} finds cross-domain samples with high causality $\bm{S}\rightarrow \bm{T}$ and constructs a causality-aware dataset under the psychological assumptions for downstream causal model.


\begin{assumption}\label{assumption: Similarity-Causality Alignment}
(Similarity-Causality Alignment).
For a given user, the preference for an item $i_s$ in the source domain causally influences its preference for a similar item $i_t$ in the target domain.
\end{assumption}
This assumption is grounded in psychological theories, such as \emph{Cognitive Consistency Theory}~\cite{kruglanski2018cognitive}, which posits that individuals tend to favor information that aligns with their pre-existing beliefs and behaviors across various domains. 
In a real-world content-sharing platform, users' interaction with a specific note causes them to be directed toward similar videos, thus maintaining cognitive harmony. 
This phenomenon is further supported by psychological principles, such as selective exposure~\cite{fischer2005selective} and confirmation bias~\cite{kassin2013forensic}.




We use two similarity measures to assess the similarity between items across two domains in practice. 
The content embedding is learned by a classification task that assigns items to shared categories, entities, etc., across two domains (e.g., music, sports, and beauty) to capture \textbf{content similarity}. 
In addition, we account for \textbf{behavioral similarity} by utilizing a graph-based encoder~\cite{he2020lightgcn} on the global user-item interaction graph to generate behavior-driven embeddings. 
The similarity between an item $i_s$ and an item $i_t$ in each space is computed as follows:
\begin{equation}
\text{sim}^{\{c,b\}}(i_s,i_t) = (\bm{e}_{i_s}^{\{c,b\}})^{\top}\bm{e}_{i_t}^{\{c,b\}},
\end{equation}
where $\bm{e}_{i}^{c}$ and $\bm{e}_{i}^{b}$ represent the content and behavior embeddings.

To ensure a balanced fusion of content and behavioral similarity, we apply a \textit{cumulative distribution function (CDF)} transformation to calibrate their scales into a comparable range. and a maximum operator to select top-ranked items within each similarity type.
The overall similarity score between $i_s$ and $i_t$ is computed as:
\begin{equation}
\text{sim}(i_s, i_t) = \max\left( \bm{F}^c(\text{sim}^c(i_s,i_t)),\ \bm{F}^b(\text{sim}^b(i_s,i_t)) \right).
\end{equation}
where $\bm{F}^{c}(\cdot)$ and $\bm{F}^{b}(\cdot)$ denote the empirical CDFs estimated from content and behavior similarities, respectively.
The calibration process is illustrated in Figure~\ref{fig:calibration}. 

Furthermore, we define the similarity between a behavior sequence $\mathcal{S}_s = (i_s^1, i_s^2, \dots, i_s^l)$ in the source domain and a target item $i_t$ in the target domain as:  
\begin{equation}
\text{sim}(\mathcal{S}_s, i_t) = \max \{\text{sim}(i_s^j, i_t)|j=1,\ldots, l \}.
\end{equation}

Leveraging this similarity measure, \ds ~ can extract causal behavior pairs from observed cross-domain behaviors.  
Specifically, for each user, we consider target item $i_t$ that occurs after the source behavior sequence $S_s^u$, and define the causal behavior pairs as:

\begin{equation}
\mathcal{D}^{+}_{u} = \{(u, \mathcal{S}_s^u, i_t) \mid \text{sim}(\mathcal{S}_s, i_t) > \tau, i_t \in \mathcal{S}_t^u\},
\end{equation}
where $\tau$ is a hyperparameter used to filter out noisy data.

By aggregating the behavior pairs from all users, we obtain the complete set of causal behavior pairs:
\begin{equation}
\mathcal{D}^{+} = \bigcup\limits_{u \in \mathcal{U}} \mathcal{D}^{+}_{u}.
\end{equation}
Thus, \ds~ constructs a causality-aware dataset $\mathcal{D}$ by assigning positive labels to behavior pairs in \(\mathcal{D}^{+}\), and negative labels otherwise. This dataset is then used for downstream causal modeling.

\subsection{\model}
In this section, we first present the causal model $f$ trained on causality-aware dataset $\mathcal{D}$, capturing how source-domain behaviors influence target-domain preference, i.e., $(\bm{L}, \bm{S}) \rightarrow \bm{I} \rightarrow \bm{T}$. Then we derive the theoretical foundations of optimizing causal model $f$ and propensity model $e$ under a Partial Label Causal Loss. The optimization promotes $f$ to recognize unseen causal patterns beyond $\mathcal{D}$,  yielding a theoretically unbiased causal modeling.

\subsubsection{Backbone model}
We employ the classic attention-based model SASRec~\cite{kang2018self} as our backbone model to learn the causality.

The concatenation of the content and ID embeddings, used in the original implementation of SASRec, forms the initial item embeddings.
Moreover, a user feature encoder is employed to capture the user-specific causal preferences.


The causal model $f$ predicts the probability of interacting with a target domain item, caused by user's source domain behaviors:
\begin{equation}
f(u, \mathcal{S}_s^u,i_t) = \sigma ([f_\text{se}(f_\text{fe}(u),f_\text{em}(i_s^{u,1}),f_\text{em}(i_s^{u,2}),\dots)]^\top f_\text{em}(i_t)),
\end{equation}
where $f_\text{em}(\cdot) \in \mathbb{R}^d$ is the item embedding function combining the content and ID embeddings.
$f_\text{fe}(\cdot) \in \mathbb{R}^d$ generates user embeddings from their source domain features.
$f_\text{se}(\cdot) \in \mathbb{R}^d$ represents the sequence encoder from SASRec, consisting of multiple self-attention~\cite{vaswani2017attention} and feed-forward layers, enhanced by residual connections~\cite{he2016deep}, layer normalization~\cite{ba2016layer}, and dropout regularization~\cite{srivastava2014dropout}.
$\sigma(\cdot)$ denotes the sigmoid function.

Binary Cross Entropy (BCE) loss is employed to train the causal model $f$ on the set of real causal samples $\mathcal{C}$:
\begin{equation}
\label{eq:causal_loss_1}
\small
  \mathcal{L} = - \left[
        \sum_{(u, S_s^u,i_t) \in \mathcal{C}}  \log \left(f(u, S_s^u,i_t) \right) +
        \sum_{(u, S_s^u,i_t) \notin \mathcal{C}} \log \left( 1 - f(u, S_s^u,i_t) \right) \right],
\end{equation}

\subsubsection{Partial Label Causal Loss}
While real causal samples $\mathcal{C}$ are unavailable, we propose the following assumption.
\begin{assumption}
(Asymmetry of Similarity and Causality).  High similarity is not always a necessary condition for causal relationships. 
Certain causal item pairs may exhibit low similarity, due to latent factors or limitations of learned item embeddings, i.e., there exists pair $(u, S_s^u,i_t) \in \mathcal{C}$ such that $(u, S_s^u,i_t) \notin \mathcal{D}^{+}$.
\end{assumption}
This assumption aligns with the common intuition in real-world recommender systems. 
Consequently, directly training on causality-aware dataset $\mathcal{D}$ labeled from \ds, yields two issues:
First, some causal behavior pairs are mistakenly labeled as negative samples.
Second, exclusive selection of highly similar samples biases causal model towards capturing labeling strategy rather than causality.

To address these issues, we propose a novel Partial Label Causal Loss, which corrects the training on the partially labeled dataset $\mathcal{D}$ by employing a causal model $f(x)$ and a propensity model $e(x)$.

The propensity score~\cite{imbens2015causal} is adopted to quantify the probability of a true causal instance $x = (u, S_s^u, i_t)$ being selected into the set of causal behavior pairs $\mathcal{D}^+$, i.e.: 
\begin{equation}
e(x) = p(x \in \mathcal{D}^{+} \mid x, x \in \mathcal{C}).
\end{equation}

The following two propositions establish the theoretical foundations of how to optimize the causal model $f(x)$ and the propensity model $e(x)$ on constructed causality-aware dataset $\mathcal{D}$.

\begin{proposition}
(Partial Label Causal Loss). 
Given $s \perp\!\!\!\perp y \mid x$ and true propensity score $e(x)$, the causal model $f(x)$ is optimized by the following loss to learn a theoretically unbiased estimation of the true causal probability $p(y=1\mid x)$, with a partially labeled dataset $\mathcal{D}$:
\begin{equation}
\begin{aligned}
\mathcal{L} & = \frac{1}{n}\sum_{x} \left[h(x)\delta_1^f(x) +(1-h(x))\delta_0^f(x) \right]
\end{aligned}
\end{equation}
where $h(x) = \text{sg}\left[s+(1-s)\frac{f(x)(1-e(x))}{1-f(x)e(x)}\right]$, $\delta_1^f(x)=-\log(f(x))$ and $\delta_0^f(x)=-\log(1-f(x))$.
$s$ and $y$ are the binary indicators $s=\mathbb{I}(x\in\mathcal{D}^{+})$ and $y=\mathbb{I}(x\in\mathcal{C})$.
$n$ is the number of all samples.
\end{proposition}
\begin{proof}
The ideal unbiased loss for the  causal model $f$ is:
\begin{equation}
\begin{aligned}
\mathcal{L} & = -\mathbb{E}_{x, y} \log p(y\mid x,f) \\ 
& = -\mathbb{E}_{x} \mathbb{E}_{y\mid x}  \log p(y\mid x,f) \\
& = -\frac{1}{n}\sum_{x} \mathbb{E}_{y\mid x}  \log p(y\mid x,f) \\
&= \frac{1}{n}\sum_{x} \left[ p(y=1|x)\delta_1^f(x) + (1-p(y=1|x))\delta_0^f(x) \right],
\end{aligned}
\end{equation}

Given only observable labels $s$ from \ds, but without access to the true label distribution $p(y=1 \mid x)$, the probability of a sample $x \notin \mathcal{D}^{+}$ being causal is:
\begin{equation}
\begin{aligned}
p(y=1 \mid x,s=0)&=\frac{p(y=1\mid x)p(s=0\mid x,y=1)}{p(s=0\mid x)}\\
    &=\frac{p(y=1\mid x)(1-p(s=1\mid x,y=1))}{1-p(s=1\mid x)}\\
    &=\frac{p(y=1\mid x)(1-p(s=1\mid x,y=1))}{1-p(y=1\mid x)p(s=1\mid x,y=1)}.
\end{aligned}
\end{equation}

Thus, the causal probability for any sample $x$ is:
\begin{equation}
\begin{aligned}
    p(y=1\mid x,s)&=sp(y=1\mid x,s=1)\ +\\
    &\quad\ (1-s)p(y=1\mid x,s=0)\\
    &= s+(1-s)p(y=1\mid x,s=0).
\end{aligned}
\end{equation}

Leveraging the current estimation $f(x)$ for $p(y=1\mid x)$ and $e(x)$ for $p(s=1\mid x,y=1)$, we derive a corrected causal label $h(x)$ as:
\begin{equation}
h(x)=p(y=1\mid x,s)=\text{sg}\left[s+(1-s)\frac{f(x)(1-e(x))}{1-f(x)e(x)}\right],
\end{equation}
where $\text{sg}[\cdot]$ denotes the stop-gradient operation, preventing the parameter updates on $f(x)$ and $e(x)$ when computing $h(x)$, since it acts solely as a corrected label.

Accordingly, the loss function is reformulated as:
\begin{equation}
\qedhere
\begin{aligned}
\mathcal{L} & = -\mathbb{E}_{x, y, s} \log p(y\mid x,f) \\ 
& = -\mathbb{E}_{x,s} \mathbb{E}_{y\mid x,s}  \log p(y\mid x,f) \\
& = -\frac{1}{n}\sum_{x} \mathbb{E}_{y\mid x,s}  \log p(y\mid x,f) \\
& = \frac{1}{n}\sum_{x} \left[ h(x)\delta_1^f(x) +(1-h(x))\delta_0^f(x) \right]
\end{aligned}
\end{equation}
\end{proof}

In practice, the corrected loss is exclusively applied to items interacted with by users in the target domain, as these interactions reflect real user preferences and are more likely to exhibit causal associations with their source domain behaviors.
Thus, $h(u, S_s^u, i_t)$ is set to $0$ for items where $i_t \notin \mathcal{S}_t^u$.
The final loss function is: 
\begin{equation}
\label{eq:causal_loss_f}
\small
\mathcal{L}_f=\sum_{u\in\mathcal{U}} \left\{\sum_{i_t\in\mathcal{S}_t^u}[h(i_t)\delta_1^f(i_t)+(1-h(i_t))\delta_0^f(i_t)]+\sum_{i_t\notin\mathcal{S}_t^u}\delta_0^f(i_t)\right\}.
\end{equation}
Here, $\mathcal{S}_s^u$ and $u$ as inputs to $h(\cdot)$ and $\delta(\cdot)$ are omitted for clarity.

In parallel to the causal model $f$, a propensity model is employed to learn the label assignment probability $e(x)$.
Consistent with the labeling strategy in \ds, we perform max-pooling over item-wise multiplication sequence followed by some feed-forward layers:
\begin{equation}
e(u,\mathcal{S}_s,i_t)=\sigma\left(e_\text{mlp}\left(\text{max-pooling}(e_\text{em}(i_s^1) \odot  e_\text{em}(i_t) ,\dots)\right)\right),
\end{equation}
where $e_\text{em}$ represents the item embedding function similar to $f_\text{em}$, $e_\text{mlp}$ denotes a multi-layer perceptron, and $\odot$ indicates the element-wise multiplication. 

In fact, if the propensity model can have arbitrary capacity, it's impossible to know if an unlabeled instance is due to low propensity or low causality.
Therefore, we design the propensity model to learn only the similarity-driven labeling mechanism, allowing the causal model to focus on the underlying causality.

\begin{proposition}
(Propensity Score Modeling). 
Given the estimated causal probability $h(x)$, the propensity model $e(x)$ is optimized by the following loss to approximate the label assignment probability $p(s=1 \mid x, y=1)$, enabling the correction of selection bias induced by the data labeling process in \ds:
\[
\mathcal{L} = \frac{1}{n}\sum_{x} h(x) [s\delta_1^e(x)+(1-s)\delta_0^e(x)].
\]
\end{proposition}
\begin{proof}
Similar to the causal model $f$, the loss function for the propensity model $e$ is derived as follows:
\begin{equation}
\qedhere
\begin{aligned}
\mathcal{L} & = -\mathbb{E}_{x, y, s} \log p(s\mid x,y,e) \\ 
& = -\frac{1}{n}\sum_{x} \mathbb{E}_{y\mid x,s}  \log p(s\mid x,y,e) \\
& = -\frac{1}{n}\sum_{x} [ h(x)\log p(s\mid x,y=1,e) \\ 
& \quad \qquad \ \quad + (1-h(x))\log p(s\mid x,y=0,e))] \\
& = \frac{1}{n}\sum_{x} h(x) [s\delta_1^e(x)+(1-s)\delta_0^e(x)]
\end{aligned}
\end{equation}
\end{proof}
In practice, the loss used for training the propensity model $e$ is:
\begin{equation}
\label{eq:causal_loss_e}
\small
\mathcal{L}_e=\sum_{u\in\mathcal{U}}\sum_{i_t\in\mathcal{S}_t^u}h(i_t)[s\delta_1^e(i_t)+(1-s)\delta_0^e(i_t)].
\end{equation}

Intuitively, the corrected causal label $h(i_t)$  serves to down-weight non-causal samples during propensity estimation, as $e(x)$ represents the labeling probability for real causal samples.

Finally, due to the specially designed architectures and loss functions,  propensity model $e(x)$ and the causal model $f(x)$ capture complementary aspects from the causal dataset $\mathcal{D}$, $e(x)$ for labeling mechanism and $f(x)$ for real causality.

\subsection{\enhance}
In this section, we describe how \es~ integrates cross-domain representation learned from \ms~ in an adaptive manner to facilitate a more comprehensive modeling of target domain user preferences, i.e., $(\bm{L}, \bm{I})\rightarrow \bm{T}$.
\paragraph{Cross-Domain Self-Attention}
Let the vector representations of user $u$ learned by \ms~ from the source domain behaviors be denoted as $\bm{r}_s^u = f_\text{se}(u, \mathcal{S}_s^u)$, and those learned from the base recommendation model in the target domain as $\bm{r}_t^u$, such that $\bm{r}_s^u, \bm{r}_t^u \in \mathbb{R}^d$.
The cross-domain self-attention mechanism treats the representations $\bm{r}_s^u$ and $\bm{r}_t^u$ as a sequence of length two, represented as $\bm{R}^u = [\bm{r}_s^u, \bm{r}_t^u] \in \mathbb{R}^{2 \times d}$. 
The self-attention is defined as follows:
\begin{equation}
[\bm{p}_s^{u}, \bm{p}_t^u] = \bm{P}^u = \text{softmax}\left(\frac{(\bm{R}^u \bm{W}_q)(\bm{R}^u \bm{W}_k)^\top}{\sqrt{d}}\right)(\bm{R}^u \bm{W}_v),
\end{equation}
where $\bm{W}_q, \bm{W}_k, \bm{W}_v \in \mathbb{R}^{d \times d}$ are learnable weights corresponding to queries, keys, and values.

This self-attention mechanism enables the exchange of information between the source and target domains, resulting in more comprehensive user preference representations  $\bm{p}_s^{u}, \bm{p}_t^u \in \mathbb{R}^d$.
\paragraph{Cross-Domain Gated Attention}
We select specific user features $\bm{c}^u \in \mathbb{R}^{d_c}$, which implicitly describe the user's behavioral consistency and preference strength across domains, such as indicators of cold-start status or relative click activity.
Guided by these features, a gated attention mechanism is introduced to adaptively weight the contributions of $\bm{p}_s^u$ and $\bm{p}_t^u$. 
The gating function is defined as:
\begin{equation}
g^u = \sigma\left(\text{LeakyReLU}\left(\bm{c}^u \bm{W}_{g,1} + \bm{b}_{g,1}\right)\bm{W}_{g,2} + \bm{b}_{g,2}\right),
\end{equation}
where $\bm{W}_{g,1} \in \mathbb{R}^{d_c \times d}$, $\bm{W}_{g,2} \in \mathbb{R}^{d \times 1}$, and $\bm{b}_{g,1}\in \mathbb{R}^d, \bm{b}_{g,2} \in \mathbb{R}$ are learnable weights. 
LeakyReLU function is used for activation.

The gate scalar $g^u$ is then used to combine the user preferences $\bm{p}_s^u$ and $\bm{p}_t^u$ obtained from the cross-domain self-attention:
\begin{equation}
\bm{\nu}^u = g^u \odot \bm{p}_s^u + (1 -g^u)\bm{p}_t^u,
\end{equation}
The resulting user representation $\bm{\nu}^u$ is designed to integrates valuable information from both domains.

Finally, the preference scores for recommendations are computed by the dot product between the integrated user representation $\bm{\nu}^u$ and the representation $\bm{r}^{i_t}$ of the target item $i_t$ from the base recommendation model. 
The recommendation loss function is:
\begin{equation}
\small
\mathcal{L}_\text{rec} = -\sum_{u \in \mathcal{U}} \left[ \sum_{i_t \in \mathcal{S}_t^u} \log \left( \sigma\left({\bm{\nu}^u}^\top \bm{r}^{i_t} \right) \right) + \sum_{i_t \notin \mathcal{S}_t^u} \log \left( 1 - \sigma\left({\bm{\nu}^u}^\top \bm{r}^{i_t} \right) \right) \right].
\end{equation}
\subsection{Online Serving}\label{sec: Online Serving}
In this section, we provide the strategies implemented to reduce the undesirable increase of online latency without sacrificing the model performance in production.

\begin{itemize}[leftmargin=*]
    \item \textbf{Two-phase training.} We employ a dual-stage training scheme, consisting of daily full-model training on data collected within a single day, followed by intra-day incremental updates initialized from the previous day's checkpoint.
    \item \textbf{Freezing dense parameters.} During intra-day incremental updates, we freeze the dense parameters within the \ms{} and only update the sparse tables, ensuring that the distribution of the learned causal representation $\bm{r}_s^u$ remains stable.
    \item \textbf{Real-time embedding cache.} We maintain a real-time cache for causal representations $\bm{r}_s^u$ of each user, during incremental updates. These cached embeddings are directly served for online inference to reduce latency.
\end{itemize}
These strategies have been deployed in production and proven effective in meeting both latency and performance expectations.


\section{Experiments}
To evaluate our method, we conduct experiments to answer the following research questions (RQs):
\begin{itemize}[leftmargin=*]
    \item \textbf{RQ1:} How does \fr{} perform compared to baselines?
    \item \textbf{RQ2:} Can \ms{} generalize beyond observed data and identify unseen causal patterns?
    \item \textbf{RQ3:} How does the learned causal representation contribute to recommendation?
    \item \textbf{RQ4:} How does each module in CE-CDR impact performance?
    \item \textbf{RQ5:} How does similarity threshold $\tau$ affect performance?
    \item \textbf{RQ6:} How does CE-CDR perform in real-world production?
\end{itemize}

\subsection{Experimental Settings}
\subsubsection{Datasets.}
\begin{table}[t]
    \centering
    \caption{Statistics of datasets.}
    \renewcommand{\arraystretch}{0.8}
    \label{tab:dataset_statistics}
    \scalebox{0.68}{
    \begin{tabular}{l|c|c|c|c|c|c}
        \toprule
        \textbf{Datasets} & \multicolumn{2}{c|}{\textbf{Douban}} & \multicolumn{2}{c|}{\textbf{Amazon}} & \multicolumn{2}{c}{\textbf{Industry}} \\
        \midrule
        \textbf{Scenarios} & \multicolumn{2}{c|}{Book-Music} & \multicolumn{2}{c|}{Movies-CDs} & \multicolumn{2}{c}{Note-Video} \\
        \midrule
        \textbf{Domains} & Book & Music & Movies & CDs & Note & Video \\
        \midrule
        \textbf{\#Users} & 2,212 & 1,820 & 123,960 & 112,395 & 188,668,764 & 186,822,365  \\
        \midrule
        \textbf{\#Shared users} & \multicolumn{2}{c|}{1,736} & \multicolumn{2}{c|}{18,547} & \multicolumn{2}{c}{167,766,638} \\
        \midrule
        \textbf{\#Items} & 95,872 & 79,878 & 50,052 & 73,713 & 65,732,897 & 35,344,556 \\
        \midrule
        \textbf{\#Interactions} & 227,251 & 179,847 & 1,697,533 & 1,443,755 & 20,479,556,858 & 9,467,087,108 \\
        \bottomrule
    \end{tabular}
    \label{tab:statistic}
    }
\end{table}
\begin{table*}[t]
    \caption{Overall recommendation performance. The best is highlighted in bold, while the second-best is underlined.}
    \centering
    \renewcommand{\arraystretch}{0.95}
    \scalebox{0.8}{
    \begin{tabular}{l|cc|cc|cc|cc|cc|cc}
        \toprule
        \multirow{2}{*}{\textbf{Domain}} & \multicolumn{4}{c|}{\textbf{Douban}} & \multicolumn{4}{c|}{\textbf{Amazon}} & \multicolumn{4}{c}{\textbf{Industry}} \\
        \cline{2-13}
        & \multicolumn{2}{c|}{Book} & \multicolumn{2}{c|}{Music} & \multicolumn{2}{c|}{Movies and TV} & \multicolumn{2}{c|}{CDs and Vinyl} & \multicolumn{2}{c|}{Note} & \multicolumn{2}{c}{Video}\\
        \hline
        
        \textbf{Metric} & \text{HR@15} & \text{NDCG@15} & 
        \text{HR@15} & \text{NDCG@15}& \text{HR@15} & \text{NDCG@15}& \text{HR@15} & \text{NDCG@15}& \text{HR@150} & \text{NDCG@150}& \text{HR@150} & \text{NDCG@150} \\
        \hline
        \textbf{Base Model} & 0.3052 & 0.2303 & 0.3196 & 0.2248 & 0.4959 & 0.3084 & 0.4514 & 0.2796 & 0.1848 & 0.0633 & 0.2238 & 0.0659 \\
        \textbf{CoNet} & 0.3961 & 0.3082 & 0.4142 & 0.3018 & 0.5437 & 0.3459 & 0.4909 & 0.3201 & 0.2265 & 0.0811 & 0.2807 & 0.0712 \\
        \textbf{MAN} & 0.4359 & 0.3433 & 0.4336 & 0.3153 & 0.5879 & 0.3822 & 0.5236 & 0.3652 & 0.2348 & 0.0821 & 0.3045 & 0.0804 \\
        \textbf{DiCUR} & 0.4475 & 0.3430 & 0.4468 & 0.3261 & 0.6337 & 0.4147 & \uline{0.5337} & 0.3749 & 0.2446 & 0.0843 & 0.3310 & 0.0851 \\
        \textbf{MiNet} & 0.4166 & 0.3209 & 0.4367 & 0.3208 & 0.5692 & 0.3666 & 0.5123 & 0.3469 & 0.2373 & \uline{0.0866} & 0.3025 & 0.0846 \\

        \textbf{TrineCDR} & \uline{0.4586} & \uline{0.3436} & \uline{0.4554} & \uline{0.3314} & \uline{0.6508} & \uline{0.4270} & 0.5284 & \uline{0.3805} & \uline{0.2484} & 0.0856 & \uline{0.3466} & \uline{0.0890} \\
        \midrule
        \textbf{\fr{}} & \textbf{0.5023} & \textbf{0.3711} & \textbf{0.4983} & \textbf{0.3600} & \textbf{0.6712} & \textbf{0.4498} & \textbf{0.5592} & \textbf{0.3914} & \textbf{0.2733} & \textbf{0.0933} & \textbf{0.3645} & \textbf{0.0992} \\
        \midrule
        
        \textbf{CoNet+CE} & 0.4907 & 0.3636 & 0.4882 & 0.3527 & 0.6552 & 0.4364 & 0.5435 & 0.3843 & 0.2660 & 0.0913 & 0.3539 & 0.0975 \\
        \textbf{MAN+CE} & 0.5094 & 0.3763 & 0.4990 & 0.3605 & 0.6689 & 0.4483 & 0.5587 & 0.3910 & 0.2780 & 0.0962 & 0.3671 & 0.0999 \\
        \textbf{DiCUR+CE} & \textbf{{0.5136}} & \textbf{{0.3784}} & {0.5026} & {0.3639} & \textbf{{0.6810}} & \textbf{{0.4592}} & \textbf{{0.5672}} & \textbf{{0.4035}} & \textbf{{0.2837}} & \textbf{{0.0987}} & \textbf{{0.3780}} & \textbf{{0.1025}} \\

        \textbf{MiNet+CE} & 0.5054 & 0.3734 & 0.4965 & 0.3587 & 0.6732 & 0.4512 & 0.5600 & 0.3920 & 0.2715 & 0.0927 & 0.3669 & 0.0999 \\

        \textbf{TrineCDR+CE} & 0.5025 & 0.3712 & \textbf{0.5031} & \textbf{0.3648} & 0.6786 & 0.4548 & 0.5657 & 0.3960 & 0.2744 & 0.0937 & 0.3627 & 0.0987 \\

        \bottomrule
    \end{tabular}
    \label{tab: overall}
    }
\end{table*}

We conduct experiments on three real-world datasets, with their statistics presented in Table \ref{tab:statistic}. 
\begin{itemize}[leftmargin=*]
\item \textbf{Douban}~\cite{zhu2020graphical}: A cross-domain recommendation dataset collected from Douban system. We use the Book and Music domains of the dataset for offline evaluation.
\item \textbf{Amazon}~\cite{ni2019justifying}: This dataset is derived from the Amazon Review Dataset. Specifically, since most subsets do not share users, we use the "Movies and TV" and "CDs and Vinyl" subsets, which share a portion of common users.
\item \textbf{Industry}: A large-scale dataset collected from a real-world content-sharing platform, Rednote (Xiaohongshu), for offline evaluation, containing user-view interactions with notes and videos.
\end{itemize}

\subsubsection{Baselines.}
We evaluated our proposed \fr{}, against several baseline approaches: CoNet, MAN, and DICUR enhance cross-domain recommendation by optimizing source domain tasks. 
In contrast, MINET and TrineCDR emphasize the utilization of source domain features.
Descriptions of baselines are provided as follows:
\begin{itemize}[leftmargin=*]
\item \textbf{Base Model} is a basic dual tower model, where both towers are implemented as a multi-layer perceptron with ReLU activations.
\item \textbf{CoNet}~\cite{hu2018conet} introduces dual cross-connections to enable bidirectional knowledge transfer between latent representations.
\item \textbf{MAN}~\cite{lin2024mixed} leverages mixed attention mechanisms and shared group embeddings for group transfer learning.
\item \textbf{DiCUR}~\cite{zhao2024discerning} learns domain-shared and domain-specific user representations through canonical correlation analysis~\cite{van2011generalized}.
\item \textbf{MiNet}~\cite{lin2024mixed} jointly models source domain, target domain, and long-term interests via hierarchical attention mechanisms.
\item \textbf{TrineCDR}~\cite{song2024mitigating} enhances knowledge transferability at the feature level, interaction level, and domain level via graph propagation and contrastive learning.
\end{itemize}

\subsubsection{Implementation Settings.}
For a fair comparison, all methods are trained via distributed computing and optimized using the Adam optimizer with default settings. 
The embedding size $d$ is set to $32$ for Douban and Amazon datasets, and $64$ for Industry dataset. 
Different batch sizes are employed for each domain, to ensure an equal batch number per epoch. 
As for the hyper-parameter of \fr{}, we tuned $\tau$ to $0.9$ for Douban and Amazon datasets, and $0.85$ for Industrial dataset.
Detailed hyper-parameter analyses are provided in section \ref{sec: hyper-parameter experiment}. 
For baseline models, we use the hyper-parameters in the original implementation whenever available; otherwise, we tune them for optimal performance.

\subsubsection{Evaluation Protocol} \label{sec: metrics}
Following the previous studies\cite{cao2022disencdr, krichene2020sampled, liu2020cross}, we adopt the widely used leave-one-out method to evaluate the recommendation performance, which reserves one latest interaction as the test item for each user. 
Hit Ratio (HR) and Normalized Discounted Cumulative Gain (NDCG) are employed as the evaluation metrics, as in existing studies\cite{zhao2023cross, xu2021expanding, zhao2022multi}. 
Specifically, HR@K measures whether the relevant item is in the top-K recommended items, while NDCG@K further considers the rank of the hit. 
We set $K$ to 15 for Douban and Amazon datasets, while $150$ for Industrial dataset, as the full-scale industrial dataset has much larger candidate items, approximately $30$ million for each recommendation.
For each dataset, one domain is designated as the source domain, while the other serves as the target domain alternatively. 
The recommendation performance on the target domain is reported.

\subsection{Overall Performance (RQ1)}
We evaluated \fr{} against several baseline cross-domain methods. Additionally, we enhance these methods by incorporating the causal representation learned by CE, which serves as a model-agnostic plugin. A summary of the overall performance is presented in Table \ref{tab: overall}. 
Several key findings are observed:
First, \fr{} consistently outperforms all baseline methods across all datasets and evaluation metrics, highlighting its general effectiveness. 
Among baseline methods, those based on source domain features outperform multi-task learning approaches.
The latter may suffer from relatively insufficient utilization of cross-domain knowledge~\cite{zhao2024discerning, zhao2023cross} and negative transfer~\cite{song2024mitigating, huang2024exit}.
Second, incorporating CE as a plugin leads to the best performance of DiCUR+CE. 
Interestingly, combining CE with multi-task learning approaches results in more substantial improvements than when paired with feature based methods. 
This may be because CE aligns more closely with feature-based methods, while its divergence from multi-task learning could contribute to greater improvements.
Moreover, CoNet+CE falls short of \fr{} alone, underscoring the issue of negative transfer. 
Finally, cross-domain methods achieve greater improvement on Douban and industrial datasets compared to the Amazon dataset, 
which may result from the limited overlap of shared users in Amazon.

\subsection{Further Analysis (RQ2, RQ3)}
\subsubsection{Distribution of Learned Causality (RQ2)}
\begin{figure}[t]
    \centering
    \begin{subfigure}{0.225\textwidth}
        \centering
    \includegraphics[width=\textwidth]{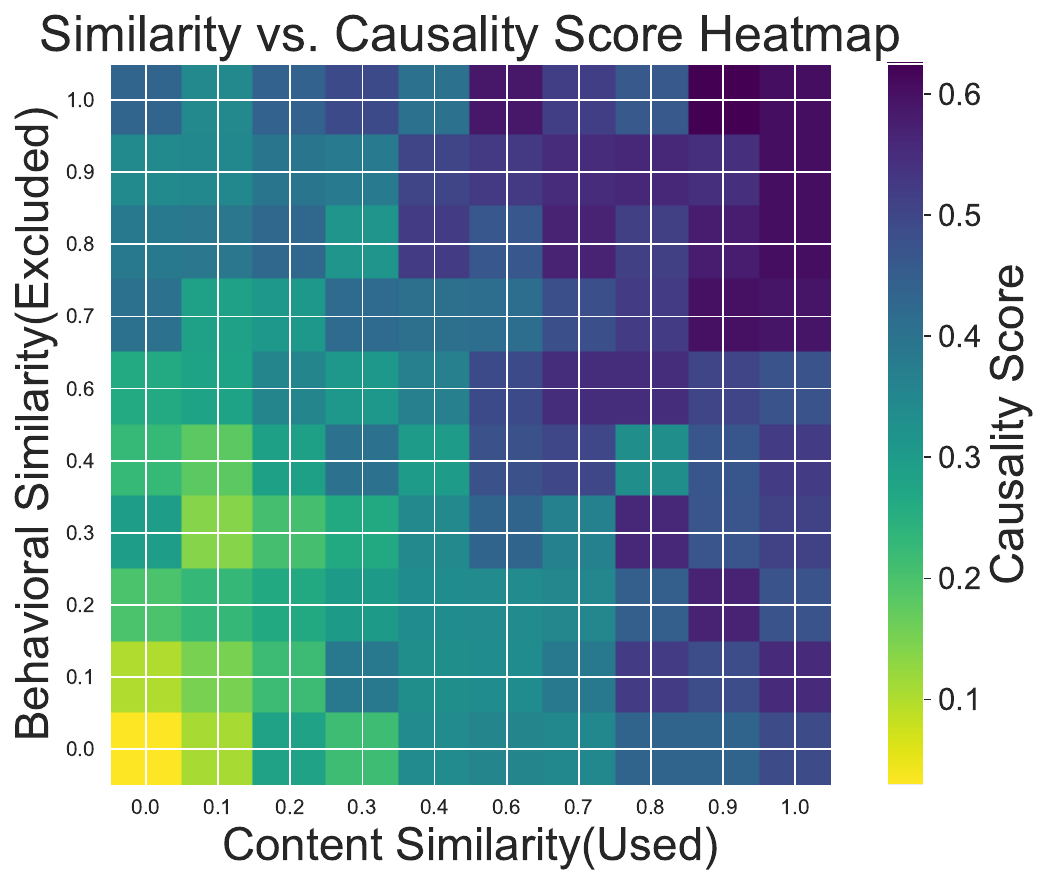}        \caption{\normalfont \scriptsize Excluding behavioral similarity from CLM.} 
    \end{subfigure}
    \begin{subfigure}{0.225\textwidth}
        \centering
\includegraphics[width=\linewidth]{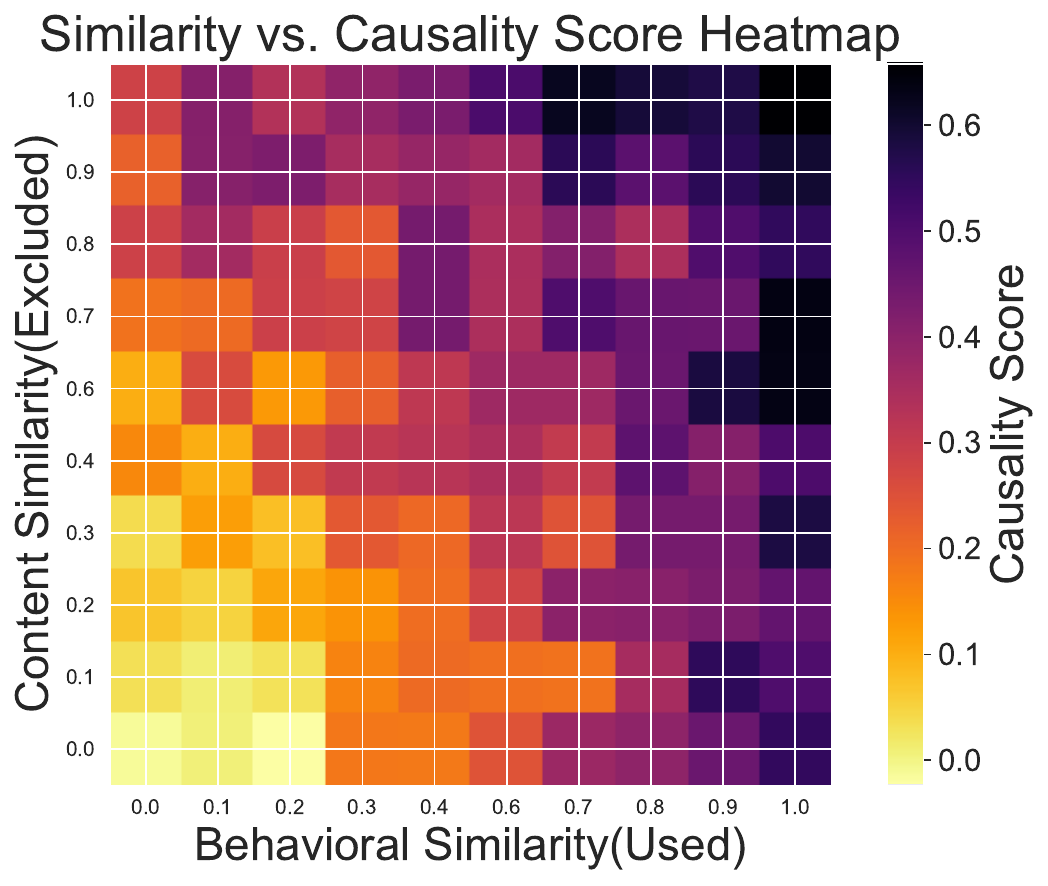}
        \caption{\normalfont \scriptsize Excluding content similarity from CLM..}
    \end{subfigure}
    \caption{Analyzing how the causality score varies with used and excluded similarities during dataset construction.} 
    \label{fig: Distribution of Learned Casuality.}
\end{figure}
In this section, we explore whether \ms{} captures unseen causal patterns beyond the causality-aware dataset.
To investigate this, we conduct an experiment on industry dataset, with \textit{video} as the target domain. 
In this setup, since real causal samples $\mathcal{C}$  are unavailable, we construct the causality-aware dataset by excluding one type of similarity, upon which \ms{} is trained. 
The average causality score $f(u, \mathcal{S}_s^u,i_t)$, for samples with different similarities $\text{sim}^{\{c,b\}}(\mathcal{S}_s^u, i_t)$, is presented as a heatmap in Figure~\ref{fig: Distribution of Learned Casuality.}.

As shown, in both experimental settings, the predicted causality scores align well with the similarity used in \ds{} as expected. 
Notably, the result also demonstrates an evident positive correlation for the excluded similarity measure. 
Although this correlation is somewhat weaker than the similarity type used for dataset construction, it indicates that \ms{} can identify unseen causality patterns. 
Moreover, several instances exhibit high causality scores, despite low similarity scores across both measures, highlighting the model's capacity to generalize and uncover underlying causal relationships beyond the causality-aware dataset.

\subsubsection{Effectiveness of Learned Causal Representation (RQ3)}
\begin{figure}[t]
    \centering
    \includegraphics[width=0.45\textwidth]{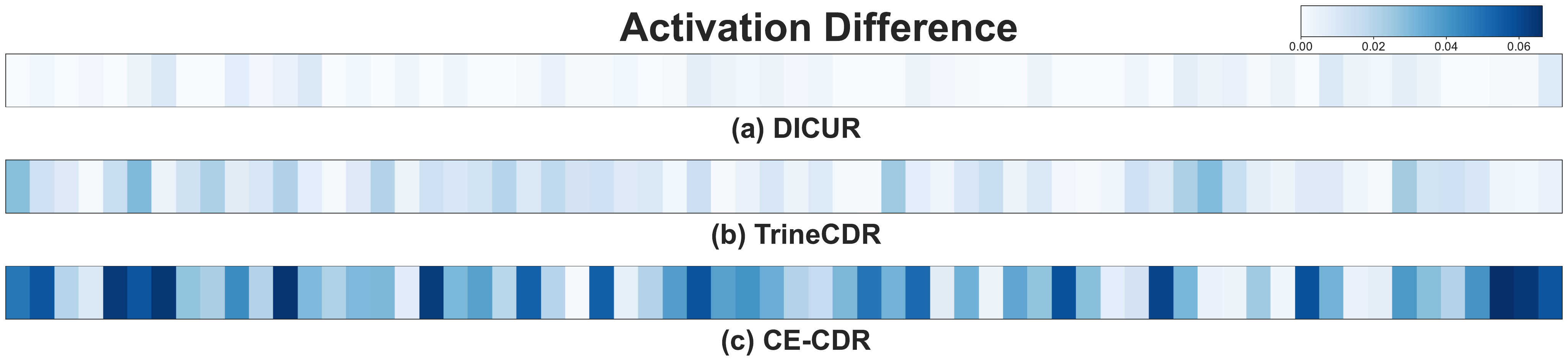}      
    \caption{Activation Difference of Cross-Domain Features.} 
    \label{fig: Submerge weight diff.}
\end{figure}
\paragraph{Activation Difference under Feature Masking}
As discussed in Section \ref{sec: introduction}, cross-domain features are often underutilized due to the sparsity of causal relationships. 
To examine the actual contribution of source domain features, we mask the source-domain behavior features and visualize the activation difference of the user representation (averaged across multiple batches) on the industry dataset, with \textit{video} as the target domain. 
Figure~\ref{fig: Submerge weight diff.} illustrates the differences before and after masking for three methods, including two representative baseline approaches when incorporating source domain features. 
Compared to both DICUR and TrineCDR, our proposed \fr{} shows a considerable activation difference, indicating that the cross-domain feature captured by \fr{} has a more significant contribution to final recommendation.

\paragraph{Dynamic Gating Interpretation}
\begin{figure}[t]
    \centering
    \includegraphics[width=0.45\textwidth]{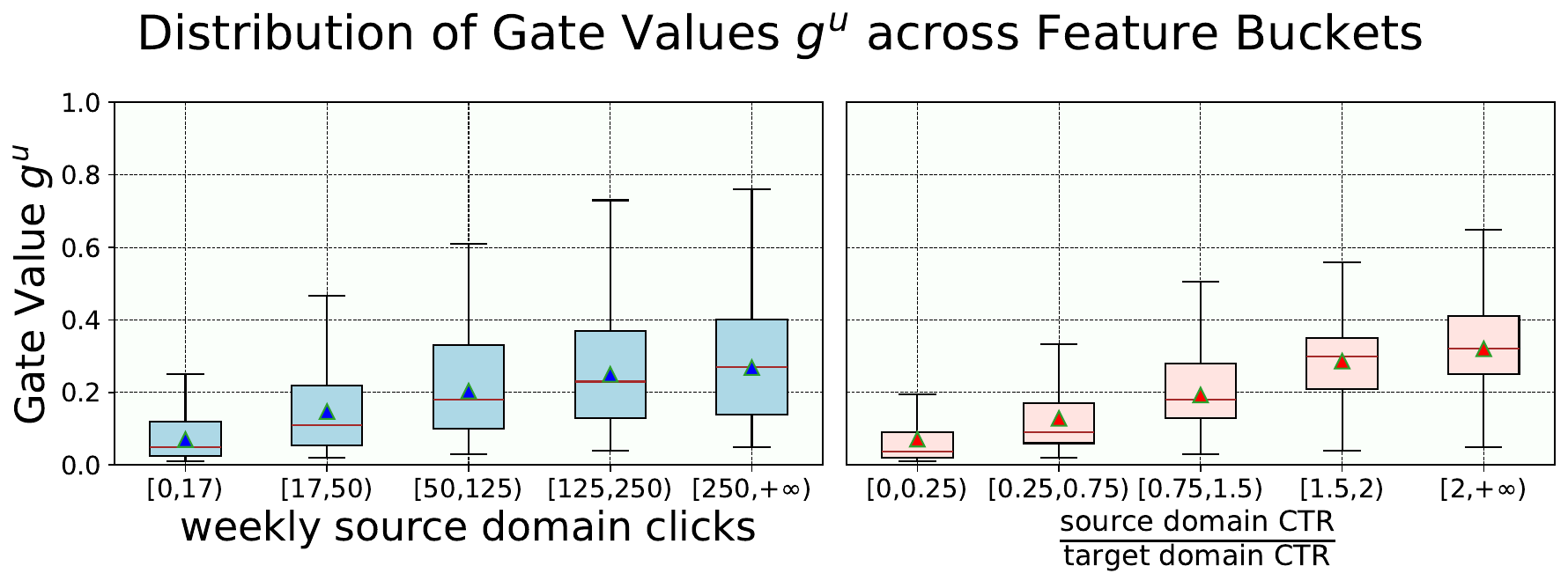}       \caption{Distributions of gate value $g^u$ for cross-domain representation fusion.} 
    \label{fig: Dynamic Gating Interpretation.}
\end{figure}
\es{} is designed based on the intuition that the impact of source-domain causal representations on target-domain recommendations should vary across users.
To assess whether the roles of $\bm{p}_s^{u}$ and $\bm{p}_t^{u}$ in cross-domain representation fusion align with our expectations, we analyze the gate value distribution $g^u$ in \es{} using box plots on the industrial dataset, with \textit{video} as target domain. 
The analyses, shown in Figure~\ref{fig: Dynamic Gating Interpretation.}, are conditioned on different buckets of two features in $\bm{c}^u$. 
Results reveal that users with richer source domain information tend to have higher gate values $g^u$, confirming the effectiveness of our dynamic gating mechanism in adaptively balancing between leveraging rich source-domain knowledge and maintaining target-domain relevance.

\subsection{Ablation Study (RQ4)}\label{sec: Ablation Study}
\begin{table}[t]
\caption{Ablation Study.}
\centering
\renewcommand{\arraystretch}{0.85}
\begin{tabular}{l|cc|cc}
\noalign{\hrule height 1.0pt}
\multirow{2}{*}{Domain} & 
\multicolumn{4}{c}{Douban} \\
\cline{2-5}
& \multicolumn{2}{c|}{Book} & \multicolumn{2}{c}{Music} \\
\hline
Metric@15 & HR  & NDCG & HR  & NDCG \\
\hline
-w/o-\ds  & 0.4357 & 0.3249 & 0.4412 & 0.3285 \\
-w/o-Cas  & 0.4412 & 0.3321 & 0.4374 & 0.3301 \\
-w/o-Att  & 0.4849 & 0.3654 & 0.4763 & 0.3479\\
-w/o-Gate & 0.4702 & 0.3561 & 0.4703 & 0.3452 \\
-w/-Cache & 0.4991 & 0.3692 & 0.4961 & 0.3589 \\
\hline
\fr{} & \textbf{0.5023} & \textbf{0.3711} & \textbf{0.4983} & \textbf{0.3600} \\
\hline
\multirow{2}{*}{Domain} &  \multicolumn{4}{c}{Industry} \\
\cline{2-5}
& \multicolumn{2}{c|}{Note} & \multicolumn{2}{c}{Video} \\
\hline
Metric@150 & HR  & NDCG & HR  & NDCG \\
\hline
-w/o-\ds  & 0.2354 & 0.0831 & 0.3193 & 0.0857 \\
-w/o-Cas  & 0.2336 & 0.0822 & 0.3212 & 0.0849 \\
-w/o-Att  & 0.2631 & 0.0902 & 0.3518 & 0.0958\\
-w/o-Gate & 0.2536 & 0.0866 & 0.3382 & 0.0920 \\
-w/-Cache & 0.2722 & 0.0923 & 0.3634 & 0.0983 \\
\hline
\fr{} & \textbf{0.2733} & \textbf{0.0933} & \textbf{0.3645} & \textbf{0.0992} \\
\hline
\end{tabular}
\label{tab: ablation}
\end{table}

We conduct an ablation study on several variants of \fr{} across both industrial dataset and the public Douban dataset. 
\textbf{-w/o-\ds} excludes \ds{}, treating all target domian interactions as causal positives. 
\textbf{-w/o-Cas} removes the partial label causal loss in \ms~and  directly treats $\mathcal{D}^+$ as the true causal set. \textbf{-w/o-Att} eliminates the cross-domain self-attention in \es{}. 
\textbf{-w/o-Gate} replaces the cross-domain gated attention in \es{} with simple concatenation.
\textbf{-w/-Cache} enables cache strategies mentioned in section \ref{sec: Online Serving}.

Table~\ref{tab: ablation} shows that the complete \fr{} consistently outperforms its variants, highlighting the importance of each component: 
(1) The labeling mechanism in \ds{} generates high-quality causal positives for downstream causal modeling.
(2) \ms{} captures generalized cross-domain causal patterns from limited positive causal samples guided by partial label causal loss. 
(3) Cross-domain self-attention in \es{} enables information exchange between domains, improving representation quality. 
(4) Cross-domain gated attention in \es{} personalizes the fusion of cross-domain representations and precisely regulates their contribution to the final recommendation.
(5)  Cache strategies for online serving eliminate latency increase with negligible performance trade-off.



\subsection{Hyperparameter Experiment (RQ5)} \label{sec: hyper-parameter experiment}
\begin{figure}[t]
    \centering
    \includegraphics[width=0.225\textwidth]{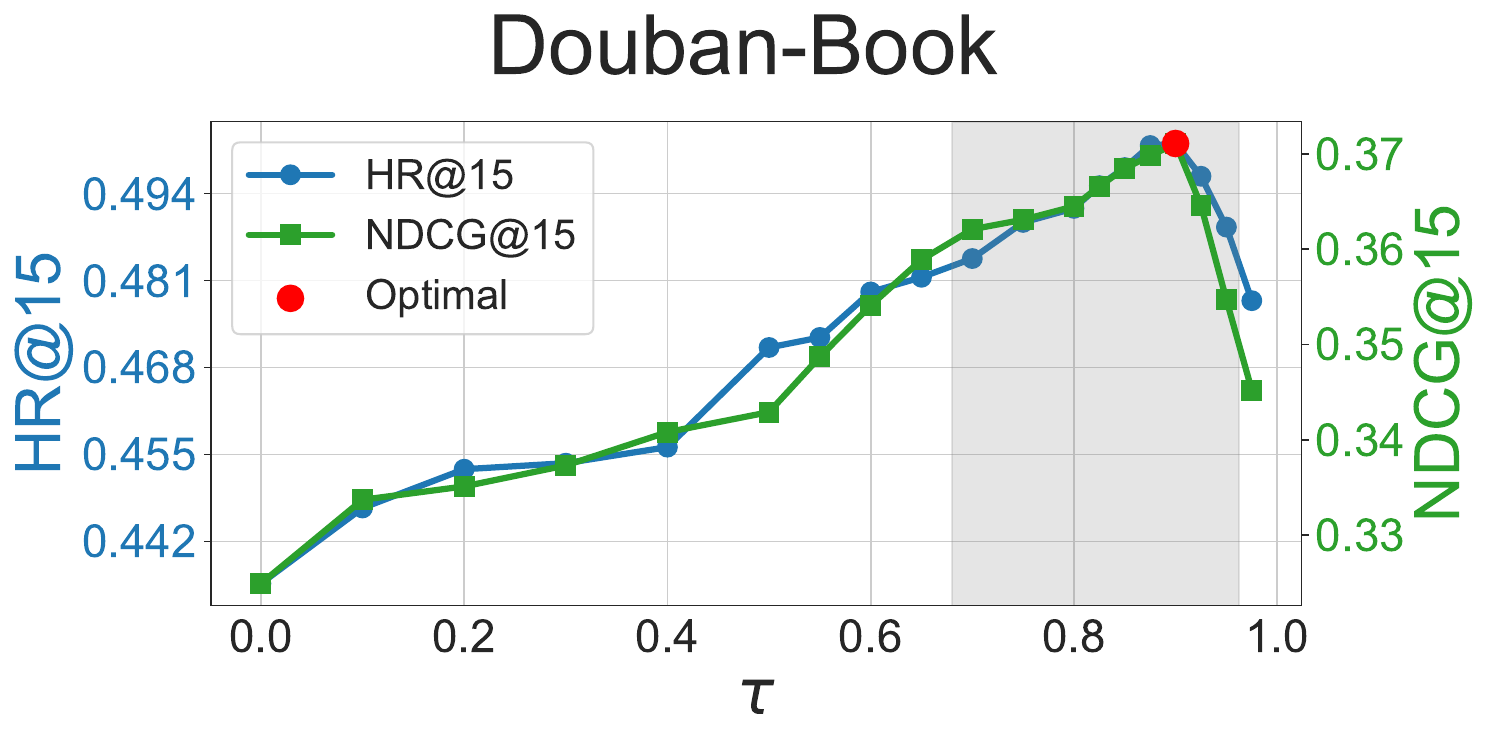}  \includegraphics[width=0.225\textwidth]{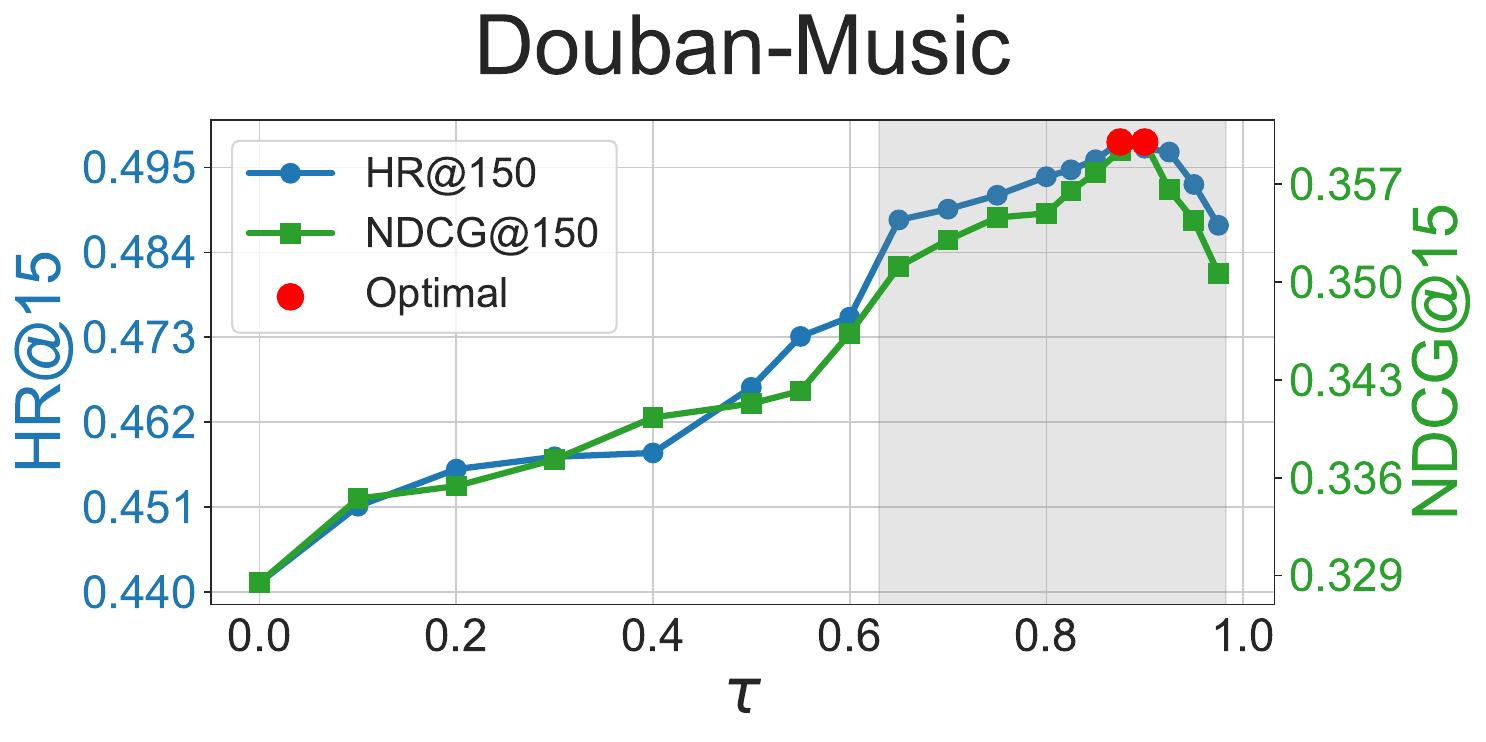} 
    \includegraphics[width=0.225\textwidth]{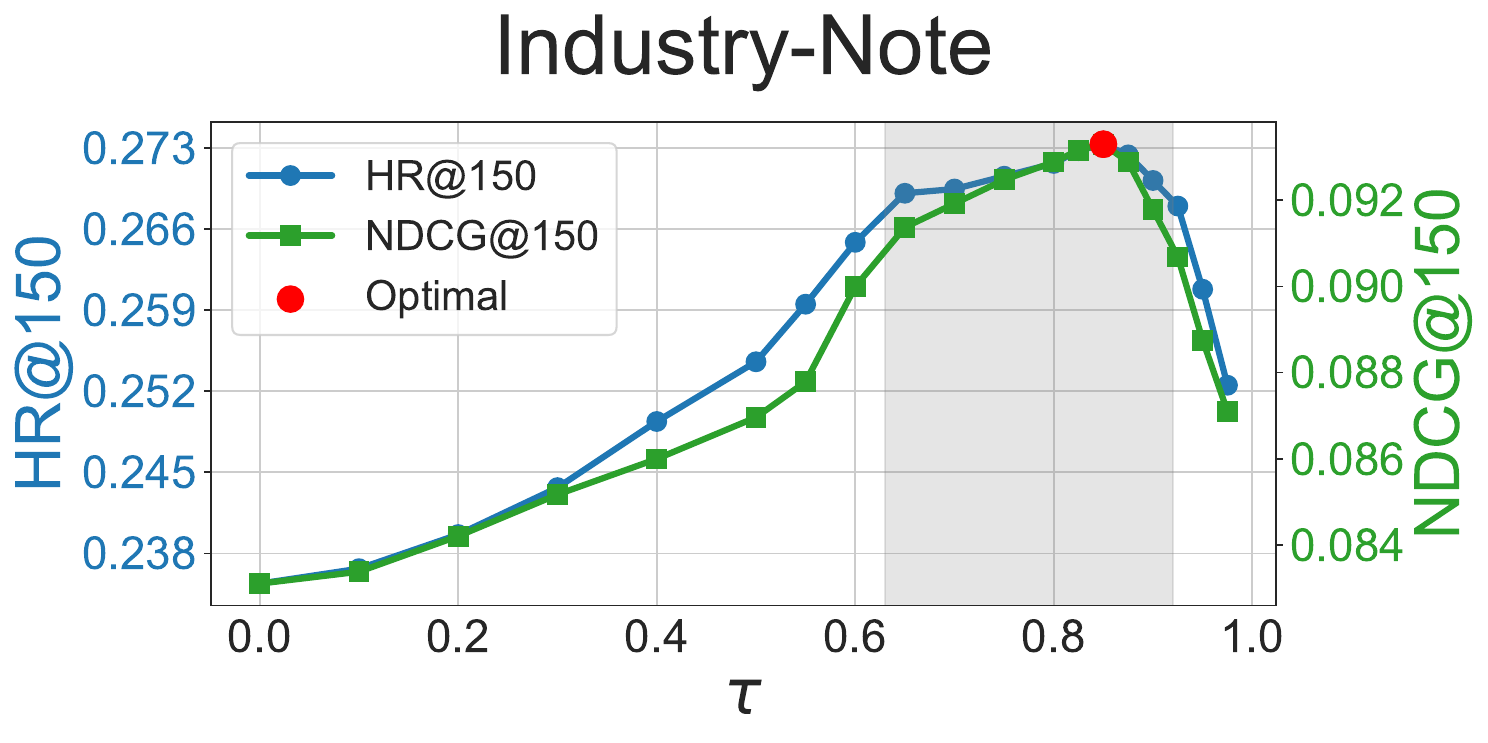}  \includegraphics[width=0.225\textwidth]{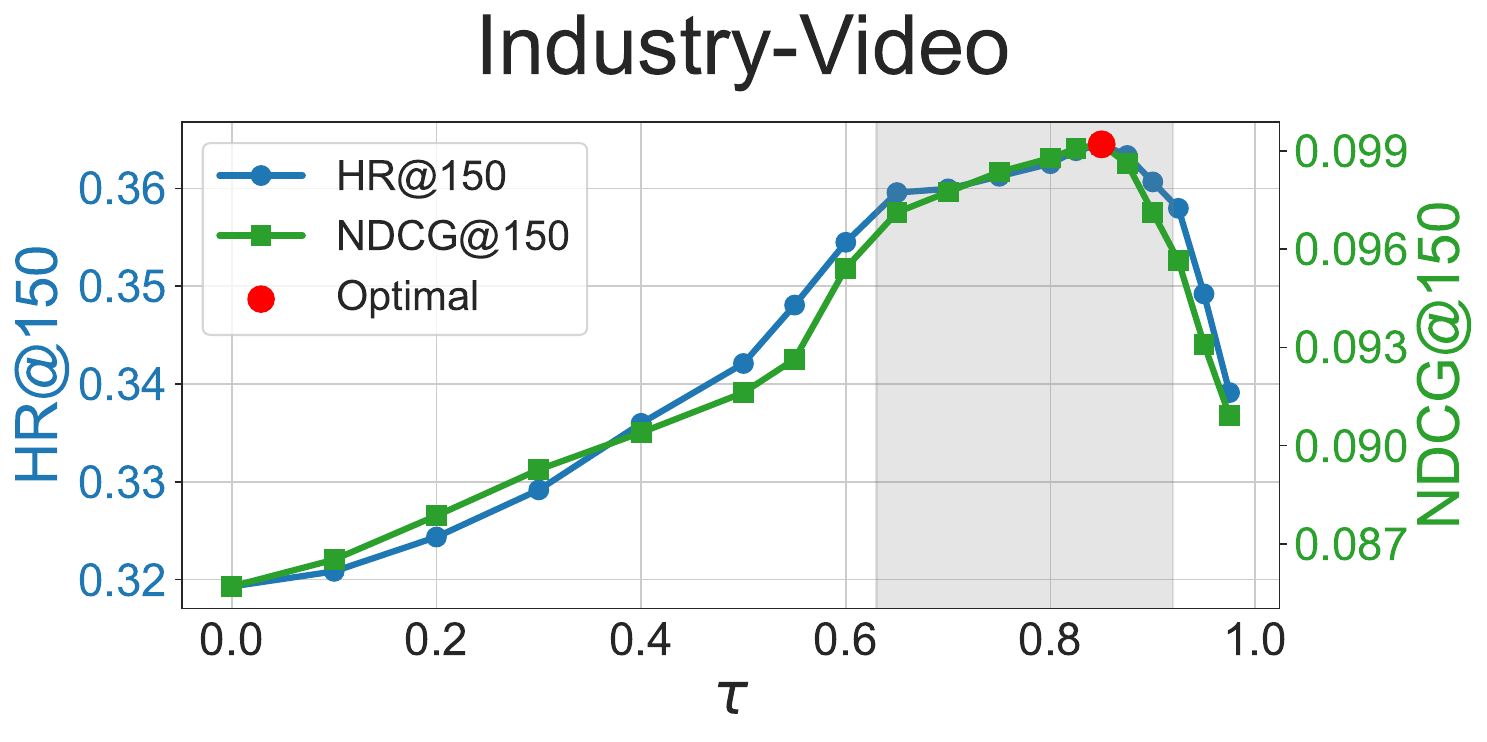} 
    \caption{Hyperparameter experiment on the similarity threshold parameter $\tau$ used in \ds{}.} 
    \label{fig: Sensity Experiment.}
\end{figure}
We investigate the sensitivity of \fr{} to the threshold parameter $\tau$ used in \ds{}.
Notably, when $\tau=0$, the model corresponds to the ablation variant CE-CDR-w/o-\ds, as discussed in Section \ref{sec: Ablation Study}.
Results in Figure~\ref{fig: Sensity Experiment.} demonstrate that selecting an optimal value of $\tau$ enhances model performance, underscoring the importance of balancing the quality and quantity of causal signals. 
A smaller $\tau$ may introduce excessive noise into the causality-aware dataset, while a larger $\tau$ overly restricts the inclusion of positive samples, limiting the amount of useful causal supervision signals.
Furthermore, the model exhibits robustness to $\tau$ values in the range of $0.65$ to $0.9$, with minor performance fluctuations, indicating that \fr{} can withstand moderate hyper-parameter variations.

\subsection{Online Experiments (RQ6)}
\subsubsection{Online Performance}
As shown in Table~\ref{Table: onlinePerformance}, a two-week A/B test conducted on a content-sharing platform, Rednote (Xiaohongshu), which serves hundreds of millions of daily active users, shows statistically significant improvements across
multiple recommendation scenarios and metrics at 95\% confidence.
The proposed \fr~ has been deployed in production since April 2025.

\subsubsection{Computational Cost}
The proposed \fr~ functions as a plugin to the online baseline model in our production environment. During training, the lightweight DCMM requires only an additional \texttt{64} CPU cores and \texttt{192} GB of memory. Such overhead is far from a bottleneck in industrial systems.
For online serving, benefiting from the real-time embedding cache as discussed in Section~\ref{sec: Online Serving}, the inference resource after incorporating \fr~ is identical to that of baseline, with both online latencies around 
$15.8$ms.

\begin{table}[t] 
\small
\centering
\caption{Online A/B test on a content-sharing platform, Rednote.}

\scalebox{0.93}{\begin{tabular}{p{1.1cm}|p{1cm}|p{0.7cm}|p{0.7cm}|p{1.2cm}|p{2.cm}}
\noalign{\hrule height 1.0pt}
\textbf{Scenario} & \textbf{Duration} & \textbf{Click} & \textbf{CTR} & \textbf{Diversity} & \textbf{Next-day Active} \\
\hline
\textbf{Video} & +0.33\% & +0.37\% & +0.12\% & +0.18\%&  +0.06\% \\
\hline
\textbf{Note} &	+0.28\%	& +0.43\% &	+0.16\%	& +0.21\% & +0.07\% \\

\hline
\end{tabular}}
\label{Table: onlinePerformance}
\end{table}

\section{Conclusion}
In this paper, we propose Causality Enhancement for Cross-Domain Recommendation (CE-CDR), a novel approach that explicitly models the cross-domain causality from both data and model perspectives. 
Leveraging principled causality labeling, and theoretically unbiased causal modeling guided by Partial Label Causal Loss, the recommendation performance is enhanced by
integrating the causal representation in an adaptive manner.
Extensive experiments demonstrate its superiority and its general applicability as a model-agnostic causal enhancement plugin. 
The well-designed architecture of its modules is also validated.
Meanwhile, \fr{} has been deployed on a real-world platform, Rednote, serving hundreds of millions of users daily. 
While our current assumption works well in practice, we aim to revisit and refine the core assumptions behind CLM for more effective causal modeling in the future.

\clearpage
\bibliographystyle{ACM-Reference-Format}
\bibliography{WWW-reference}
\appendix

\end{document}